\documentclass[12pt,reqno]{article}
\usepackage[letterpaper,margin=1.6cm]{geometry}
\usepackage{latexsym, amsfonts, amsthm,amssymb,amscd,amsmath,makeidx}
\usepackage{enumerate}
\usepackage[normalem]{ulem}
\usepackage{exscale}
\usepackage{overpic} 
\usepackage{color} 
\usepackage{graphicx}
\usepackage{bm}
\usepackage{comment}
\usepackage{caption}
\usepackage{float}
\usepackage{array} 
\usepackage{booktabs} 
\usepackage{tabularx}
\usepackage{cancel}
\usepackage{mathrsfs}
\usepackage{amssymb}
\usepackage{mathtools}

\usepackage{tikz}
\usepackage{blindtext}
\usepackage{bbm}
\usepackage{fancyhdr} 
\usepackage{mathrsfs}
\usepackage{wrapfig}
\usepackage{hyperref}
\usepackage{bookmark}
\usepackage{cleveref}
\usepackage{cases}
\usepackage{multirow}
\usepackage{pdflscape}
\usepackage[title]{appendix}

\hypersetup{ colorlinks, citecolor=red, filecolor=black, linkcolor=blue,
	urlcolor=black }

\definecolor{DarkGreen}{rgb}{0.2,0.6,0.2}

\newcolumntype{Y}{>{\centering\arraybackslash}X}

\newcolumntype{C}{>{\centering\arraybackslash}X}

 \def\ignore#1{}

\def\bR{{\mathbb R}}

  \def\bN{\mathbb N} 
\def\bT{\mathbb T}

\numberwithin{equation}{section}

\usetikzlibrary{positioning,calc,cd}

\allowdisplaybreaks[4]

  \def\cE{{\mathscr E}}
  
\def\cL{{\mathscr L}}

\newtheorem{theorem}{Theorem}[section]
\newtheorem{proposition}[theorem]{Proposition}

\theoremstyle{definition}
\newtheorem{definition}[theorem]{Definition}
\newtheorem{example}[theorem]{Example}

 \def\<{\langle} \def\>{\rangle} 

\begin{document}

\title{Universal portfolios in continuous time:\\ an approach in pathwise It\^o calculus}

\author{Xiyue Han$^*$ and Alexander Schied\thanks{ University of Waterloo,
 200 University Ave W, Waterloo, Ontario, N2L 3G1, Canada. E-Mails: {\tt
 xiyue.han@uwaterloo.ca, aschied@uwaterloo.ca}.\hfill\break
The authors gratefully acknowledge support from the Natural Sciences and
 Engineering Research Council of Canada through grant
 RGPIN-2017-04054.  They are grateful to David Pr\"omel for pointing out the work \cite{AllanCuchieroLiuProemel}. }} \date{  \normalsize  First version: April 16, 2025\\
 this version: July 11, 2026}

\maketitle
\begin{abstract}
We provide a simple and straightforward approach to a continuous-time version of Cover's universal portfolio strategies within the model-free context of F\"ollmer's pathwise It\^o calculus. We establish the existence of the universal portfolio strategy and prove that its portfolio value process is the average of all values of constant rebalanced strategies. This result relies on a systematic comparison between two alternative descriptions of self-financing trading strategies within pathwise It\^o calculus. We moreover provide a comparison result for the performance and the realized volatility and variance of constant rebalanced portfolio strategies.
\end{abstract}


\section{Introduction}

In mainstream finance, the price evolution of a risky asset is usually modeled as a stochastic process defined on some probability space. It is therefore subject to model uncertainty. In some situations, however, it is possible to construct continuous-time strategies on a path-by-path basis and without making any probabilistic assumptions on the asset price evolution. For instance, such a model-free approach is possible for constant-proportion portfolio insurance (CPPI), which the second author investigated while  discussing the paper \cite{ZagstKraus} with Rudi Zagst at TU M\"unchen. The results of this investigation were eventually published in \cite{SchiedCPPI}. The present paper is a continuation of the strictly pathwise approach taken in \cite{SchiedCPPI} and some later follow-up work such as \cite{SchiedSpeiserVoloshchenko,SchiedVoloshchenkoArbitrage}. The pathwise, model-free approach to continuous-time finance rests on the seminal  work of F\"ollmer \cite{FoellmerIto,FoellmerECM} and its application to option pricing by Bick and Willinger \cite{BickWillinger} and Lyons \cite{Lyons95}. Recent work on model-free finance includes, for instance,  Bartl et al.~\cite{BartlKupperProemel}, Chiu and Cont \cite{ChiuCont},  Karatzas and Kim \cite{KaratzasKim}, and Allan, Cuchiero, Liu, and Pr\"omel \cite{AllanCuchieroLiuProemel}.

In Section~\ref{portfolio strategies section}, we carry out a systematic comparison between two possible descriptions of continuous-time trading strategies in F\"ollmer's pathwise It\^o calculus. On the one hand, there are self-financing trading strategies that describe the number of shares of each asset that are held at time $t$. On the other hand, there are portfolio strategies that describe what percentage of the available capital is currently invested in each asset. Proposition~\ref{portfolio prop}, our first result, states that both descriptions are basically equivalent under a no-short-sales constraint. Examples~\ref{const port ex} and~\ref{CPPI ex} discuss constant rebalanced portfolio strategies and CPPI, respectively. 

Section~\ref{universal portfolio section}
 contains a construction of a continuous-time version of the universal portfolio strategy, which was introduced by Cover in \cite{Cover}. Since Cover's original construction in discrete time is model-free, it is natural to attempt a similar construction within the model-free framework of pathwise It\^o calculus. Earlier approaches to this question were  given by Jamshidian \cite{Jamshidian1992} in a setup in which prices are modeled by It\^o processes and by Cuchiero, Schachermayer, and Wong \cite{CuchierEtAl}, where a  general class of Cover-type strategies was obtained by averaging over certain functionally generated portfolios. Subsequently, Allan, Cuchiero, Liu, and Pr\"omel \cite{AllanCuchieroLiuProemel} constructed pathwise versions of Cover's universal portfolio for sets of admissible portfolios, which are controlled paths in the rough path sense. Here, we restrict ourselves to  Cover's original idea and the simple setting of F\"ollmer's pathwise It\^o calculus, which avoids the intricacies of rough path theory. In this simple framework, we provide a straightforward construction of the continuous-time universal portfolio strategy by averaging over all constant rebalanced strategies. Our  main result is Theorem~\ref{universal portf thm}, which establishes the existence of the universal portfolio strategy and identifies its portfolio value process as the average of all values of constant rebalanced strategies.

 In Section~\ref{cb portfolio section}, we investigate further properties of constant rebalanced portfolio strategies in our model-free context of pathwise It\^o calculus. Part \ref{const rebal prop thm 1} of Theorem~\ref{const rebal prop thm} shows that a constant rebalanced portfolio strategy beats the weighted geometric mean of the underlying assets, a formula that can be regarded as an It\^o
analogue of the classical inequality between the weighted arithmetic and
geometric means. Theorem~\ref{const rebal prop thm} also  provides upper bounds on the realized variance and the realized volatility of the constant rebalanced portfolio strategy. These bounds are given in terms of the weighted averages of the realized variance and volatility of the underlying assets, respectively. 

For the convenience of the reader, some notation, terminology, and preliminary results  on pathwise It\^o  calculus are collected in Section~\ref{Ito section}.
 
\section{Portfolio strategies as self-financing trading strategies}\label{portfolio strategies section}

Throughout this paper, we work in a strictly pathwise setting as in \cite{SchiedCPPI,SchiedSpeiserVoloshchenko}. We refer to those papers and to Section \ref{Ito section}
 for terminology, notation, and background on pathwise It\^o calculus as introduced by F\"ollmer  \cite{FoellmerIto}. In this setting, the prices of $d$ liquidly traded assets are modeled as continuous functions $S^1,\dots,S^d:[0,\infty)\to(0,\infty)$. These assets may consist of stocks, bonds, currencies, commodity futures, derivatives, etc.  We assume furthermore that we are given a refining sequence of partitions $(\bT_n)_{n\in\bN}$ along which asset prices are observed. In a financial context, it is natural to  think of the partitions  $\bT_n$ as corresponding to ever refining data sets, e.g.,  from yearly to quarterly to monthly and  all the way down to millisecond and  tick intervals. We assume that the trajectories $S^i$ possess continuous quadratic variations $\<S^i\>$ and covariations $\<S^i,S^j\>$, so that $\bm S:=(S^1,\dots, S^d)\in QV^d$ (see Section \ref{Ito section} for this terminology). Here and in the sequel, we use bold font to denote vectors, and \lq\lq$\cdot$" will denote the inner product between two vectors.
 
 A \emph{trading strategy} is a function $\bm\xi:[0,\infty)\to\bR^d$ that is an admissible integrand for $\bm S$ in the sense of Definition~\ref{adm def}. It is called \emph{self-financing} if its value process $V_t:=\bm\xi\cdot\bm S_t=\sum_i\xi^i_tS^i_t$ satisfies
 $$V_t=V_0+\int_0^t\bm\xi_s\,d\bm S_s,\qquad t\ge0,
 $$
 where the integral on the right-hand side is the F\"ollmer integral (see Theorem~\ref{FoellmerThm}). Let 
 $$\Delta:=\Big\{\bm x=(x^1,\dots,x^d):x^i\ge0\text{ for }i=1,\dots,d\text{ and }\sum_{i=1}^dx^i=1\Big\}
 $$
 denote the $(d-1)$-dimensional standard simplex. A \emph{portfolio strategy} is a function $\bm\pi:[0,\infty)\to\Delta$ that is an admissible integrand for the It\^o logarithm $\cL(\bm S)$ as defined in Definition~\ref{DD def} and \eqref{vector cL eq}. The \emph{portfolio value process} of $\bm\pi$ is defined as 
 \begin{equation}\label{Vpi eq}
 V^{\bm\pi}_t:=\cE\bigg(\int_0^\cdot\bm\pi_s\,d\cL(\bm S)_s\bigg)_t,\qquad t\ge0,
 \end{equation}
 where $\cE(X)$ denotes the Dol\'eans--Dade exponential of $X\in QV$ (see Definition~\ref{DD def}).
 
 \begin{proposition}\label{portfolio prop} There exists a one-to-one correspondence between all portfolio strategies $\bm\pi$ and all self-financing strategies $\bm\xi$ that satisfy the no-short-sales condition $\xi^i_t\ge0$ for all $i$ and $t$ and whose value process satisfies $V_0=1$ and $V_t>0$ for all $t$. With the notation 
 $$\frac{\bm \pi_t}{\bm S_t}:=\Big(  \frac{\pi^{1}_t}{S^{1}_t},\ldots,\frac{\pi^{d}_t}{S^{d}_t} \Big),$$
  it is given by the identity
 \begin{equation}\label{pi xi eq}
 \frac{\bm \pi_t}{\bm S_t}=\frac{\bm\xi_t}{V_t},\qquad  t\ge0,
 \end{equation}
 and in this case we have 
 $$V_t=V^{\bm \pi}_t=\cE\bigg(\int_0^\cdot\frac{\bm\pi_s}{\bm S_s}\,d\bm S_s\bigg)_t,\qquad t\ge0.$$
  \end{proposition}
 
 \begin{proof} If a portfolio strategy $\bm\pi$ is given, then the assertion follows from Lemma 2.5 in \cite{SchiedSpeiserVoloshchenko} (note that it is assumed in \cite{SchiedSpeiserVoloshchenko} that one of the assets is a money market account; this assumption can easily be dropped if needed). 
 
 Now assume that a self-financing strategy $\bm\xi$ is given that satisfies the no-short-sales condition $\xi^i_t\ge0$ for all $i$ and $t$ and whose value process satisfies $V_0=1$ and $V_t>0$ for all $t$. Let $T>0$ be given. By Definition \ref{adm def}, there exists $n\in\bN$,  open sets $U\subset\mathbb R^d$ and   $W\subset\mathbb R^n$,  a function $f\in C^{2,1}(U\times W)$, and  $\bm  A\in CBV([0,T],W)$ such that $\bm S_t\in U$ and  $\bm\xi_t=\nabla_{ \bm  x}f(\bm S_t,\bm A_t)$ for $0\le t\le T$. By Theorem \ref{FoellmerThm},
 \begin{align*}
 V_t&=V_0+\int_0^t\bm\xi_s\,d\bm S_s=f(\bm S_t,\bm A_t)+B_t, \end{align*}
 where $B$ is a sum of Riemann--Stieltjes integrals with integrators $A^k$ and $\<S^i,S^j\>$, and hence $(\bm A,B) \in CBV([0,T],W\times\bR)$.  Let
 $$O:=\{(\bm x,\bm a,b)\in U\times W\times\bR:f(\bm x,\bm a)+b>0\}.$$
Then $O$ is an open set containing $(\bm S_t,\bm A_t,B_t)$ for $0\le t\le T$ as $V_t > 0$, and, if we define $g(\bm x, \bm a,b):=\log( f(\bm x,\bm a)+b)$ on $O$, then $g\in C^{2,1}(O)$ and we have
  $$\nabla_{\bm x}g(\bm x,\bm a,b)=\frac{\nabla_{\bm x}f(\bm x,\bm a)}{f(\bm x,\bm a)+b}.
 $$
 Hence, if we define $\bm\pi$ via \eqref{pi xi eq}, then 
 $$ \frac{\bm \pi_t}{\bm S_t}=\frac{\bm\xi_t}{V_t}=\frac{\nabla_{\bm x}f(\bm S_t,\bm A_t)}{f(\bm S_t,\bm A_t)+B_t}=\nabla_{\bm x}g(\bm S_t,\bm A_t,B_t)$$
  is an admissible integrand for $\bm S$. It now follows from \eqref{log Ito eq} and the associativity formula for the F\"ollmer integral in \cite[Theorem 13]{SchiedCPPI} that $\bm\pi$ is an admissible integrand for $\cL(\bm S)$. Finally, the formula for $V$ and $V^{\bm\pi}$ now follows from \eqref{Vpi eq},  \eqref{log Ito eq}, and associativity.
 \end{proof}

\begin{example}\label{const port ex} A \emph{constant rebalanced portfolio strategy} is a portfolio strategy that is constant in time. We will henceforth identify the strategy $\bm\pi$ with the constant value it takes in the $(d-1)$-dimensional standard simplex $\Delta$. A special case is the popular \emph{equal-weight strategy} $\bm\pi^{\text{eq}}=(1/d,\dots,1/d)$. The portfolio value process of a constant rebalanced portfolio strategy $\bm\pi\in\Delta$ is given by
\begin{equation}\label{portf value for const reb strat}
V_t^{\bm\pi}=\cE\big(\bm\pi\cdot(\cL(\bm S)_\cdot-\cL(\bm S)_0)\big)_t.
\end{equation}
We will further analyze constant rebalanced portfolio strategies in Section~\ref{cb portfolio section}.
\end{example}

\begin{example}[CPPI]\label{CPPI ex} Constant Proportion Portfolio Insurance (CPPI) was first studied by Perold \cite{Perold}, Black and Jones \cite{BlackJones}, and Black and Perold \cite{BlackPerold}. It provides a strategy that yields superlinear participation in future asset returns while retaining a guarantee for part of the invested capital (\lq\lq the floor"). 
It is typically analyzed in a two-asset market in which $S:=S^1$  is a risky asset and $B:=S^2$ is a money market account. That is, 
$
B_t=\exp(\int_0^tr_s\,ds)$,
where $r:[0,\infty)\to\bR$ is measurable and satisfies $\int_0^t|r_s|\,ds<\infty$ for all $t>0$. In this situation, one is looking for a self-financing trading strategy, or equivalently for a possibly leveraged weight process $\bm\pi$ with $\pi^1_t+\pi^2_t=1$, not necessarily taking values in $\Delta$, that satisfies the following two conditions:
\begin{enumerate}
\item  The portfolio value $V^{\bm\pi}_t$ should never fall below the floor $\alpha B_t$, which one would have attained by investing the fraction $\alpha $ of the initial capital into the money market account right from the start. 
\item A leveraged multiple $m>0$ of the \lq\lq cushion", $C_t:=V^{\bm\pi}_t-\alpha B_t$, should be invested into the risky asset. That is, $\pi^1_t=mC_t/V^{\bm\pi}_t$, while the remaining fraction $\pi^2_t=1-\pi^1_t$ will be held in the money market account. 
\end{enumerate}
For $m>1$, the money-market weight $\pi_t^2$ may be negative.
It was shown in \cite{SchiedCPPI} that this problem can be solved in our present pathwise setting by first letting 
$$C_t=(1-\alpha)\Big(\frac{S_t}{S_0}\Big)^mB_t^{1-m}e^{-\frac{1}2m(m-1)\<\log S\>_t}
$$
and then solving for $\bm\pi$ from the identity $V^{\bm\pi}_t=C_t+\alpha B_t$, the formula \eqref{Vpi eq} for $V^{\bm\pi}_t$, and the It\^o differential equation \eqref{DD SDE eq} satisfied by $V^{\bm\pi}_t$. A comparison with options-based portfolio insurance (OBPI) was conducted by Zagst and Kraus \cite{ZagstKraus}.
\end{example}

\section{The universal portfolio strategy}\label{universal portfolio section}

Universal portfolios were introduced in discrete time by Cover \cite{Cover}.  Jamshidian \cite{Jamshidian1992} provided a continuous-time approach in a setup in which prices are modeled by It\^o processes. Cuchiero et al.~\cite{CuchierEtAl} introduced a general class of Cover-type strategies in continuous time by averaging over classes of functionally generated portfolios.  Subsequently, Allan et al.~\cite{AllanCuchieroLiuProemel} constructed pathwise versions of Cover's universal portfolio for sets of admissible portfolios, which are controlled paths in the rough path sense. Here, we restrict ourselves to  Cover's original idea and the simple setting of F\"ollmer's pathwise It\^o calculus, which avoids the intricacies of rough path theory. In this framework, we provide a straightforward construction  of universal portfolios in continuous time that is based on constant rebalanced strategies and is thus a direct continuous-time analogue of the discrete-time case in  \cite{Cover}.  The following theorem   establishes the universal portfolio strategy as an average of constant rebalanced strategies, weighted by the performance of their portfolio value processes.

\begin{theorem}\label{universal portf thm} For a Borel probability measure $\mu$ on $\Delta$, we define 
$$\widehat V_t:=\int_{\Delta}V^{\bm\pi}_t\,\mu(d\bm\pi).
$$
Then 
$$\widehat{\bm\pi}_t:=\frac{\int_{\Delta}\bm\pi V^{\bm\pi}_t\,\mu(d\bm\pi)}{\widehat V_t}
$$
is a portfolio strategy with portfolio value process $\widehat V$. It is called the {\bf universal portfolio strategy}. 
\end{theorem}

\begin{proof} Let us write $\bm L_t:=\cL(\bm S)_t-\cL(\bm S)_0\in QV^d$. By \eqref{portf value for const reb strat},  the portfolio value process of the constant rebalanced portfolio strategy $\bm\pi\in\Delta$ is
$$V^{\bm\pi}_t=\cE(\bm\pi\cdot\bm L)_t=\exp\bigg(\bm\pi\cdot \bm L_t-\frac12\sum_{i,j=1}^d\pi^i\pi^j\<L^i,L^j\>_t\bigg).
$$
Now we define a function $f:\bR^d\times\bR^{d\times d}\to\bR$ by
$$f(\bm x,\bm a):=\log\int_\Delta e^{\bm\pi\cdot\bm x-\frac12\bm\pi^\top\bm a\bm\pi}\,\mu(d\bm\pi).
$$
By dominated convergence, $f$ is smooth on its domain and 
$$\nabla_{\bm x}f(\bm x,\bm a)=\frac1{e^{f(\bm x,\bm a)}}\int_\Delta \bm\pi e^{\bm\pi\cdot\bm x-\frac12\bm\pi^\top\bm a\bm\pi}\,\mu(d\bm\pi).
$$
Hence, for the matrix-valued $CBV$-function $\bm A$ with components  $A^{ij}_t:=\<L^i,L^j\>_t$, 
$$\nabla_{\bm x}f(\bm L_t,\bm A_t)=\frac{\int_\Delta \bm\pi V^{\bm\pi}_t\,\mu(d\bm\pi)}{\int_\Delta V^{\bm\pi}_t\,\mu(d\bm\pi)}=\widehat{\bm\pi}_t
$$
is an admissible integrand for $\bm L$ and in turn for $\cL(\bm S)$. Therefore $\widehat{\bm\pi}$ is a portfolio strategy. 

Let $V^{\widehat{\bm\pi}}$ be the portfolio value process of $\widehat{\bm\pi}$. We must show that $V^{\widehat{\bm\pi}}=\widehat V$. We know that $V^{\widehat{\bm\pi}}_0=1=\widehat V_0$ and that 
\begin{equation}\label{widehat pi value process SDE eq}
dV^{\widehat{\bm\pi}}_t= V^{\widehat{\bm\pi}}_t\widehat{\bm\pi}_t\,d\bm L_t,
\end{equation}
where we have used \eqref{Vpi eq}, \eqref{DD SDE eq}, and the associativity of the F\"ollmer integral  \cite[Theorem 13]{SchiedCPPI}. If we can show that $\widehat V$ satisfies the same It\^o differential equation~\eqref{widehat pi value process SDE eq} as $V^{\widehat{\bm\pi}}$, then the assertion will follow from uniqueness of solutions to this equation, which was established  in \cite[Theorem 2.12]{MishuraSchied}. To this end, we introduce
$$g(\bm x,\bm a):=e^{f(\bm x,\bm a)}=\int_\Delta e^{\bm\pi\cdot\bm x-\frac12\bm\pi^\top\bm a\bm\pi}\,\mu(d\bm\pi).
$$
Then 
$\widehat V_t=g(\bm L_t,\bm A_t)$.  Hence, by Theorem~\ref{FoellmerThm},
\begin{align*}
d\widehat V_t&=\nabla_{\bm x}g(\bm L_t,\bm A_t)\,d\bm L_t+\sum_{i,j=1}^dg_{a^{ij}}(\bm L_t,\bm A_t)\,dA^{ij}_t+\frac12\sum_{i,j=1}^dg_{x^i,x^j}(\bm L_t,\bm A_t)\,d\<L^i,L^j\>_t.
\end{align*}
We have
$$g_{a^{ij}}(\bm x,\bm a)=-\frac12\int_\Delta \pi^i\pi^je^{\bm\pi\cdot\bm x-\frac12\bm\pi^\top\bm a\bm\pi}\,\mu(d\bm\pi)=-\frac12 g_{x^i,x^j}(\bm x,\bm a).
$$
Together with the fact that $A^{ij}=\<L^i,L^j\>$, we arrive at 
$$d\widehat V_t=\nabla_{\bm x}g(\bm L_t,\bm A_t)\,d\bm L_t=\widehat V_t\widehat{\bm\pi}_t\,d\bm L_t,
$$
which is the desired  It\^o differential equation~\eqref{widehat pi value process SDE eq}. 
\end{proof}

\section{Properties of constant rebalanced portfolio strategies}\label{cb portfolio section}

In this section, we take a look at the performance of constant rebalanced portfolio strategies as introduced in  Example~\ref{const port ex}. 
The following theorem compares  the portfolio value of such a strategy   with the weighted geometric average of the underlying assets. We look at portfolio growth, realized variance, and realized volatility, and find that the constant rebalanced portfolio strategy outperforms the geometric average in the following senses. This explains, for instance, the underperformance of geometric stock market indices such as the Value Line Geometric Index. Part \ref{const rebal prop thm 1} of the following theorem can be regarded as an It\^o analogue of the classical inequality between the weighted arithmetic and geometric means. 
 
 \begin{theorem}\label{const rebal prop thm} The portfolio value process $V_t^{\bm\pi}$ of a  constant rebalanced portfolio strategy $\bm\pi\in\Delta$ has the following properties. 
 \begin{enumerate}
 \item\label{const rebal prop thm 1} For each $t$, the performance of $\bm \pi$ beats the $\bm\pi$-weighted geometric mean of the performances of the underlying assets:
$$V_t^{\bm\pi}\ge \prod_{i=1}^d\bigg(\frac{S^i_t}{S^i_0}\bigg)^{\pi^i}.
 $$
 \item\label{const rebal prop thm 2} The asymptotic growth rate of $V_t^{\bm\pi}$  exceeds the $\bm\pi$-weighted average asymptotic growth rate of the underlying assets. More precisely, 
 $$\limsup_{t\uparrow\infty}\frac1t\log V_t^{\bm\pi}\ge\limsup_{t\uparrow\infty}\frac1t\sum_{i=1}^d\pi^i\log S^i_t
$$
and
$$\liminf_{t\uparrow\infty}\frac1t\log V_t^{\bm\pi}\ge\liminf_{t\uparrow\infty}\frac1t\sum_{i=1}^d\pi^i\log S^i_t.
$$
\item\label{const rebal prop thm 3}  The realized variance  of $ V^{\bm\pi}$ is lower than the $\bm\pi$-weighted average realized variance  of the underlying assets:
$$\<\log  V^{\bm\pi}\>_t\le\sum_{i=1}^d\pi^i\<\log S^i\>_t.
$$
\item\label{const rebal prop thm 4} The realized volatility of $ V^{\bm\pi}$ is lower than the $\bm\pi$-weighted average realized volatility  of the underlying assets:
$$\sqrt{\<\log  V^{\bm\pi}\>_t}\le\sum_{i=1}^d\pi^i\sqrt{\<\log S^i\>_t}.
$$
 \end{enumerate}
 
  \end{theorem}
 
\begin{proof} Let $\bm X\in QV^d$ be defined by $X^i_t:=\log S^i_t$. 

\ref{const rebal prop thm 1} Using the definition of the It\^o logarithm, we can write:\begin{align*}
V_t^{\bm\pi}&=\cE\big(\bm\pi\cdot(\cL(\bm S)-\cL(\bm S)_0)\big)_t\\
&=\exp\bigg(\bm\pi\cdot(\bm X_t-\bm X_0)+\frac12\bigg(\sum_{i=1}^d\pi^i\<X^i\>_t-\sum_{i,j=1}^d\pi^i\pi^j\<X^i,X^j\>_t\bigg)\bigg)\\
&=\prod_{i=1}^d\bigg(\frac{S^i_t}{S^i_0}\bigg)^{\pi^i}
\exp\bigg(\frac12\bigg(\sum_{i=1}^d\pi^i\<X^i\>_t-\sum_{i,j=1}^d\pi^i\pi^j\<X^i,X^j\>_t\bigg)\bigg).\end{align*}
We show now that the exponent on the right-hand side is nonnegative. This will establish \ref{const rebal prop thm 1}. To this end,
consider the $d\times d$-matrix $C=(c_{ij})_{i,j=1,\dots, d}$ with entries $c_{ij}=\<X^i,X^j\>_t$. It is easy to see  that $C$ is symmetric and nonnegative definite. Thus, the function $\varphi:\bR^d\to\bR$ defined by 
$\varphi(\bm x):=\bm x^\top C\bm x$ is convex. Let
$\bm e_i=(0,\dots,0,1,0,\dots,0)$
denote the $i^{\text{th}}$ unit vector. 
Then the convexity of $\varphi$ gives
\begin{align*}
\sum_{i=1}^d\pi^i\<X^i,X^i\>_t&=\sum_{i=1}^d\pi^ic_{ii}=\sum_{i=1}^d\pi^i\bm e_i^\top C\bm e_i=\sum_{i=1}^d\pi^i\varphi(\bm e_i)\\
&\ge \varphi\Big(\sum_{i=1}^d \pi^i\bm e_i\Big)=\varphi(\bm\pi)=\bm\pi^\top C\bm\pi\\
&=\sum_{i,j=1}^d\pi^i\pi^j\<X^i,X^j\>_t.
\end{align*}
This completes the proof of \ref{const rebal prop thm 1}.

\ref{const rebal prop thm 2} This follows immediately by taking logarithms in \ref{const rebal prop thm 1}.

\ref{const rebal prop thm 3} Taking logarithms in our formula for $V_t^{\bm\pi}$ obtained in the proof of \ref{const rebal prop thm 1} yields
\begin{align*}
\log V_t^{\bm\pi}=\bm\pi\cdot(\bm X_t-\bm X_0)+\frac12\bigg(\sum_{i=1}^d\pi^i\<X^i\>_t-\sum_{i,j=1}^d\pi^i\pi^j\<X^i,X^j\>_t\bigg).
\end{align*}
By \eqref{A BV eq}, constant terms and functions of bounded variation do not contribute to the quadratic variation of a trajectory. Hence, $\<\log  V^{\bm\pi}\>_t=\<\bm\pi\cdot\bm X\>_t$. Letting $C$, $\varphi$, and $\bm e_i$ be as in the proof of part \ref{const rebal prop thm 1}, we obtain,
\begin{align*}
\<\log  V^{\bm\pi}\>_t&=\<\bm\pi\cdot\bm X\>_t=\sum_{i,j=1}^d\pi^i\pi^j\<X^i,X^j\>_t=\varphi(\bm\pi)\le\sum_{i=1}^d\pi^i\varphi(\bm e_i)\\
&=\sum_{i=1}^d\pi^i\<X^i,X^i\>_t=\sum_{i=1}^d\pi^i\<\log S^i\>_t.
\end{align*}

\ref{const rebal prop thm 4} Let $\psi(\bm x):=\sqrt{\varphi(\bm x)}$. Then $\psi$ is a seminorm and hence convex. We may thus argue similarly as in the proof of part \ref{const rebal prop thm 3}:
\begin{align*}
\sqrt{\<\log  V^{\bm\pi}\>_t}&=\sqrt{\<\bm\pi\cdot\bm X\>_t}=\sqrt{\sum_{i,j=1}^d\pi^i\pi^j\<X^i,X^j\>_t}=\psi(\bm\pi)\le\sum_{i=1}^d\pi^i\psi(\bm e_i)\\
&=\sum_{i=1}^d\pi^i\sqrt{\<X^i,X^i\>_t}=\sum_{i=1}^d\pi^i\sqrt{\<\log S^i\>_t}.
\end{align*}
This completes the proof.
\end{proof}

\section{Preliminaries on pathwise It\^o calculus}\label{Ito section}
 
 A \emph{partition} of the time axis $[0,\infty)$ is a set $\mathbb T=\{t_0,t_1,\dots\}$ of time points such that $0=t_0<t_1<\cdots$ and  $\lim_{n\uparrow\infty}t_n=+\infty$. By $|\mathbb T|:=\sup_{i\ge0}|t_{i+1}-t_i|$ we denote the mesh of $\mathbb T$. If a partition is fixed, it will be convenient to denote by $t'$ the successor of $t\in\mathbb T$, i.e.,
 $t'=\min\{u\in\mathbb T:u>t\}$.
  In the sequel, we fix a \emph{refining sequence of partitions} $(\mathbb T_n)_{n\in\mathbb N}$ of partitions of $[0,\infty)$. That is, each $\mathbb T_n$ is a partition of $[0,\infty)$ and we have $\mathbb T_1\subset\mathbb T_2\subset\cdots$ and $\lim_{n\uparrow\infty}|\mathbb T_n|=0$. 

A continuous trajectory $ X:[0,\infty)\to\mathbb R$ is said to admit the continuous quadratic variation $\langle X\rangle$ along $(\mathbb T_n)_{n\in\mathbb N}$ if 
$$\<X\>_t:=\lim_{n\uparrow\infty}\sum_{\substack{ 
s\in\bT_n\\ s'\le t}}(X_{s'}-X_s)^2
$$
exists for each $t\ge0$ and if $t\mapsto\<X\>_t$ is continuous. In this case, we write $X\in QV$. A vector-valued continuous trajectory  $\bm X:[0,\infty)\to\bR^d$ belongs to $QV^d$ if $X^i+X^j\in QV$ for all $i,j=1,\dots, d$. Then, in particular, $X^i\in QV$ and the covariation
$$\<X^i,X^j\>_t:=\frac12\big(\<X^i+X^j\>_t-\<X^i\>_t-\<X^j\>_t\big)
$$
is a well-defined continuous function of $t$ and of locally bounded total variation, because the quadratic variations on the right-hand side are continuous and nondecreasing functions of $t$. Moreover, if $X\in QV$ and $A\in C[0,\infty)$ is of locally bounded total variation, then both $A$ and $X+A$ belong to $QV$ and
\begin{equation}\label{A BV eq}
\<A\>_t=0\quad\text{and}\quad \<X,A\>_t=0\quad\text{for all $t$;}
\end{equation}
see \cite[Remark 8]{SchiedCPPI}.

The class  $C^{2,1}(U\times W)$, for open sets $U\subset\bR^d$ and $W\subset\bR^n$, consists of all functions $f(\bm x,\bm a)$ that are twice continuously differentiable in $\bm x\in U$ and continuously differentiable in $\bm a\in W$.  We will write $f_{a^k}$ for the partial derivative of $f$ with respect to the $k^{\text{th}}$ coordinate of the vector $\bm a=(a^1,\dots, a^n)$ and $f_{x^i}$ and $f_{x^ix^j}$ for the first and second partial derivatives with respect to the components $x^i$ and $x^j$ of $\bm x$. With $\nabla_{\bm x}f=(f_{x^1},\dots, f_{x^d})$ we denote the gradient of $f$ with respect to $\bm x$. Finally, the class of continuous functions from $[0,\infty)$ to a set $V\subset\bR^n$ that are locally of bounded variation will be denoted by $CBV([0,\infty),V)$.

\begin{theorem}[F\"ollmer \cite{FoellmerIto}]\label{FoellmerThm}  Suppose  that $\bm X\in QV^d$, that $\bm A\in CBV([0,\infty),\bR^n)$, and that $f\in C^{2,1}(\bR^d\times\bR^n)$. Then 
\begin{align*}
f(\bm X_t,\bm A_t)-f(\bm X_0,\bm A_0)&=\sum_{k=1}^n\int_0^tf_{a^k}(\bm X_s,\bm A_s)\,dA_s^k+\int_0^t\nabla_{\bm x}f(\bm X_s,\bm A_s)\,d\bm X_s\\
&\qquad+\frac12\sum_{i,j=1}^d\int_0^tf_{x^ix^j}(\bm X_s,\bm A_s)\,d\<X^i,X^j\>_s,
\end{align*}
where the integrals with integrators $A^k$ and $\<X^i,X^j\>$ are taken in the usual sense of Riemann--Stieltjes integrals and the {\bf pathwise It\^o integral} $\int_0^t\nabla_{\bm x}f(\bm X_s,\bm A_s)\,d\bm X_s$ is given by the following limit of nonanticipative Riemann sums,
\begin{equation}\label{pathwise Ito integral}
\int_0^t\nabla_{\bm x}f(\bm X_s,\bm A_s)\,d\bm X_s=\lim_{n\uparrow\infty}\sum_{\substack{ 
s\in\bT_n\\ s'\le t}}\nabla_{\bm x}f(\bm X_{s},\bm A_{s})\cdot(\bm X_{s'}-\bm X_{s}).
\end{equation}
\end{theorem}

The pathwise It\^o integral \eqref{pathwise Ito integral}   will also be called the \emph{F\"ollmer integral} in the sequel. Based on the preceding theorem, we can now introduce the following class of admissible integrands for the pathwise It\^o integral.

\begin{definition}\label{adm def}
A function $ \bm\xi:[0,\infty)\to\bR^d$ is called an \emph{admissible integrand for $\bm X\in QV^d$} if for each $T>0$ there exists $n\in\bN$,  open sets $U\subset\mathbb R^d$ and   $W\subset\mathbb R^n$,  a function $f\in C^{2,1}(U\times W)$, and  $\bm  A\in CBV([0,T],W)$ such that $\bm X_t\in U$ and  $\bm\xi_t=\nabla_{ \bm  x}f(\bm X_t,\bm A_t)$ for $0\le t\le T$.
\end{definition}

If $\bm\xi$ is an admissible integrand, then we can define the F\"ollmer integral $\int_0^t\bm\xi_s\,d\bm X_s$ via \eqref{pathwise Ito integral}. That is, 
$$\int_0^t\bm\xi_s\,d\bm X_s=\lim_{n\uparrow\infty}\sum_{\substack{ 
s\in\bT_n\\ s'\le t}}\bm\xi_s\cdot(\bm X_{s'}-\bm X_{s}).
$$
Since the right-hand side only involves the values $\bm\xi_s=\nabla_{ \bm  x}f(\bm X_s,\bm A_s)$, the F\"ollmer integral $\int_0^t\bm\xi_s\,d\bm X_s$ is independent of the choice of~$f$. 

\begin{definition}\label{DD def}
Let $X,Y\in QV$ and suppose that $Y_t>0$ for all $t$.
\begin{enumerate}
\item The \emph{Dol\'eans--Dade exponential} of $X$ is given by
$$\cE(X)_t:=\exp\Big(X_t-\frac12\<X\>_t\Big).
$$
\item The \emph{It\^o logarithm} of $Y$ is defined as
$$\cL(Y)_t:=\log Y_t+\frac12\<\log Y\>_t.
$$
\end{enumerate}
\end{definition}

In the context of the preceding definition, \eqref{A BV eq} implies  that
$$\cE(\cL(Y))_t=e^{\log Y_t}=Y_t
$$
and 
$$\cL(\cE(X))_t=\log e^{X_t-\frac12\<X\>_t}+\frac12\<X\>_t=X_t.
$$
Thus, $\cE$ and $\cL$ are inverses of one another. Moreover, the pathwise It\^o formula in Theorem \ref{FoellmerThm} yields that, for $Z_t:=\cE(X)_t$,
\begin{equation}\label{DD SDE eq}
Z_t=Z_0+\int_0^tZ_s\,dX_s,
\end{equation}
or in shorthand notation $dZ_t=Z_t\,dX_t$, and 
\begin{equation}\label{log Ito eq}
\cL(Y)_t=\log Y_0+\int_0^t\frac1{Y_s}\,dY_s.
\end{equation}
If $\bm Y=(Y^1,\dots, Y^d)\in QV^d$ and $Y^i_t>0$ for all $i$ and $t$, we write
\begin{equation}\label{vector cL eq}
\cL(\bm Y)=\big(\cL(Y^1),\dots,\cL(Y^d)\big).
\end{equation}


\bibliographystyle{plain}
\bibliography{CTbook}


\end{document}